\newif\ifprocs
\procstrue
\procsfalse 

\ifprocs
\documentclass[a4paper,UKenglish]{lipics}
 
\else
\documentclass[11pt]{article}
\usepackage[letterpaper,hmargin=1.0in,vmargin=1.0in]{geometry}
\fi

\usepackage{amssymb,amsmath}
\usepackage{amsthm,sectsty}

\ifprocs
\else
\usepackage{ifpdf}
\ifpdf    
\usepackage[pdftex,colorlinks,linkcolor=blue,citecolor=blue,filecolor=blue,urlcolor=blue]{hyperref}
\else    
\usepackage[pdftex,colorlinks,linkcolor=blue,citecolor=blue,filecolor=blue,urlcolor=blue]{hyperref}
\fi
\fi
 
\usepackage{cleveref}
\usepackage{color}
\usepackage[boxed]{algorithm}
\usepackage{aliascnt}
 
\ifprocs
\else
\newtheorem{theorem}{Theorem}[section]
\Crefname{theorem}{Theorem}{Theorems}

\newaliascnt{lemma}{theorem}
\newtheorem{lemma}[lemma]{Lemma}
\aliascntresetthe{lemma}
\Crefname{lemma}{Lemma}{Lemmas}

\newaliascnt{corollary}{theorem}

\aliascntresetthe{corollary}
\Crefname{corollary}{Corollary}{Corollaries}

\newaliascnt{fact}{theorem}

\aliascntresetthe{fact}
\Crefname{fact}{Fact}{Facts}

\newaliascnt{definition}{theorem}
\newtheorem{definition}[definition]{Definition}
\aliascntresetthe{definition}
\Crefname{definition}{Definition}{Definitions}

\newaliascnt{remark}{theorem}

\aliascntresetthe{remark}
\Crefname{remark}{Remark}{Remarks}

\newaliascnt{exercise}{theorem}

\aliascntresetthe{exercise}
\Crefname{exercise}{Exercise}{Exercises}

\newaliascnt{example}{theorem}

\aliascntresetthe{example}
\Crefname{example}{Example}{Examples}

\newaliascnt{notation}{theorem}

\aliascntresetthe{notation}
\Crefname{notation}{Notation}{Notations}

\newaliascnt{problem}{theorem}

\aliascntresetthe{problem}
\Crefname{problem}{Problem}{Problems}
\fi

\theoremstyle{plain}

\newaliascnt{conjecture}{theorem}
\newtheorem{conjecture}[conjecture]{Conjecture}
\aliascntresetthe{conjecture}
\Crefname{conjecture}{Conjecture}{Conjectures}

\newaliascnt{question}{theorem}
\newtheorem{question}[question]{Question}
\aliascntresetthe{question}
\Crefname{question}{Question}{Questions}

\newaliascnt{claim}{theorem}
\newtheorem{claim}[claim]{Claim}
\aliascntresetthe{claim}
\Crefname{claim}{Claim}{Claims}

\newaliascnt{proposition}{theorem}
\newtheorem{proposition}[proposition]{Proposition}
\aliascntresetthe{proposition}
\Crefname{proposition}{Proposition}{Propositions}

\ifprocs
\newaliascnt{lemma}{theorem}
\newtheorem{mylemma}[lemma]{Lemma}
\aliascntresetthe{lemma}
\Crefname{mylemma}{Lemma}{Lemmas}
\else
\newaliascnt{mylemma}{theorem}
\newtheorem{mylemma}[lemma]{Lemma}
\aliascntresetthe{lemma}
\Crefname{mylemma}{Lemma}{Lemmas}
\fi

\newcommand{\norm}[1]{\lVert#1\rVert}

\def\E{\mathbb E}

\newcommand{\R}{\mathbb R}

\newcommand\polylog[1]{\ensuremath{\mathrm{polylog}\left(#1\right)}}

\renewcommand\epsilon{\varepsilon}

\providecommand{\eqdef}{:=}

\providecommand{\card}[1]{\lvert#1\rvert}
\providecommand{\norm}[1]{\lVert#1\rVert}

\newcommand\tO{\ensuremath{\tilde O}}

\title{Towards Resistance Sparsifiers}

\ifprocs
\author[1]{Michael Dinitz\footnote{Work supported in part by NSF award \#1464239.}}
\author[2]{Robert Krauthgamer\footnote{Work supported in part by a US-Israel BSF grant \#2010418
and an Israel Science Foundation grant \#897/13.}}
\author[3]{Tal Wagner}
\affil[1]{Johns Hopkins University, Baltimore, MD.\\ \texttt{mdinitz@cs.jhu.edu}}
\affil[2]{Weizmann Institute of Science, Rehovot, Israel.\\ \texttt{robert.krauthgamer@weizmann.ac.il}}
\affil[3]{Massachussetts Institute of Technology, Cambridge, MA.\\ \texttt{talw@mit.edu}}
\authorrunning{M. Dinitz and R. Krauthgamer and T. Wagner} 

\Copyright{Michael Dinitz and Robert Krauthgamer and Tal Wagner}


\serieslogo{}
\volumeinfo
  {Billy Editor and Bill Editors}
  {2}
  {Conference title on which this volume is based on}
  {1}
  {1}
  {1}
\EventShortName{}
\DOI{10.4230/LIPIcs.xxx.yyy.p}
\else
\author{Michael Dinitz\thanks{Johns Hopkins University, Baltimore, MD. Work supported in part by NSF award \#1464239. Email: \texttt{mdinitz@cs.jhu.edu}}
\and Robert Krauthgamer%
\thanks{Weizmann Institute of Science, Rehovot, Israel. 
Work supported in part by a US-Israel BSF grant \#2010418
and an Israel Science Foundation grant \#897/13.
Email: \texttt{robert.krauthgamer@weizmann.ac.il}
} 
\and Tal Wagner\thanks{Massachussetts Institute of Technology, Cambridge, MA.  Email: \texttt{talw@mit.edu}}
}
\fi

\begin{document}
\maketitle

\begin{abstract}
We study \emph{resistance sparsification} of graphs, in which the goal is to find a sparse subgraph (with reweighted edges) that approximately preserves the effective resistances between every pair of nodes. 
We show that every dense regular expander admits 
a $(1+\epsilon)$-resistance sparsifier of size $\tilde O(n/\epsilon)$, 
and conjecture this bound holds for all graphs on $n$ nodes. 
In comparison, spectral sparsification is a strictly stronger notion
and requires $\Omega(n/\epsilon^2)$ edges even on the complete graph.

Our approach leads to the following structural question on graphs: Does every dense regular expander contain a sparse regular expander as a subgraph? Our main technical contribution, which may of independent interest, is a positive answer to this question in a certain setting of parameters. Combining this with a recent result of von Luxburg, Radl, and Hein~(JMLR, 2014) leads to the aforementioned resistance sparsifiers.
\end{abstract}

\section{Introduction}

Compact representations of discrete structures are of fundamental importance,
both from an applications point of view and from a purely mathematical perspective. 
Graph sparsification is perhaps one of the simplest examples: 
given a graph $G(V, E)$, is there a subgraph that represents $G$ truthfully, say up to a small approximation? 
This notion has had different names in different contexts, 
depending on the property that is being preserved: 
preserving distances is known as a \emph{graph spanner}~\cite{PS89}, preserving the size of cuts is known as a \emph{cut sparsifier}~\cite{BK96}, while preserving spectral properties is known as a \emph{spectral sparsifier}~\cite{ST04a}. 
These concepts are known to be related, for example, 
every spectral sparsifier is clearly also a cut sparsifier,
and spectral sparsifiers can be constructed by an appropriate sample of
spanners \cite{KP12}.

Our work is concerned with sparsification that preserves \emph{effective resistances}.  
We define this in Section~\ref{sec:results}, 
but informally the effective resistance between two nodes $u$ and $v$ 
is the voltage differential between them when we regard the graph 
as an electrical network of resistors with
one unit of current injected at $u$ and extracted at $v$. 
Effective resistances are very useful in many applications 
that seek to cluster nodes in a network
(see \cite{vLRH14} and references therein for a comprehensive list),
and are also of fundamental mathematical interest.
For example, they have deep connections to random walks on graphs 
(see~\cite{Lov96} for an excellent overview of this connection).  
Most famously, the commute time between two nodes $u$ and $v$ (the expected time for a random walk starting at $u$ to hit $v$ plus the expected time for a random walk starting at $v$ to hit $u$) is exactly $2m$ times the effective resistance between $u$ and $v$,
where throughout $n\eqdef\card V$ and $m\eqdef\card E$.
Hence, we are concerned with sparsification which preserves commute times.

We ask whether graphs admit a good \emph{resistance sparsifier}: 
a reweighted subgraph $G'(V,E',w')$ in which the effective resistances 
are equal, up to a $(1+\epsilon)$-factor, to those in the original graph.
The short answer is yes, because every $(1+\epsilon)$-spectral sparsifier 
is also a $(1+\epsilon)$-resistance sparsifier.  
Using the spectral-sparsifiers of~\cite{BSS12}, we immediately conclude 
that every graph admits a $(1+\epsilon)$-resistance sparsifier 
with $O({n}/{\epsilon^2})$ edges.  

Interestingly, the same $1/\epsilon^2$ factor loss appears even 
when we interpret ``sparsification'' far more broadly.  
For example, a natural approach to compressing the effective resistances is to use a metric embedding (instead of looking for a subgraph): 
map the nodes into some metric, and use the metric's distances 
as our resistance estimates.  
This approach is particularly attractive since it is well-known that 
effective resistances form a metric space which embeds isometrically 
into $\ell_2$-squared (i.e., the metric is of negative type, see e.g.~\cite{DL97}).  
Hence, using the Johnson-Lindenstrauss dimension reduction lemma, 
we can represent effective resistances up to a distortion of $(1+\epsilon)$ using vectors of dimension $O(\epsilon^{-2} \log n)$, 
i.e., using total space $\tO({n}/{\epsilon^2})$.  
In fact, this very approach was used by~\cite{SS11} to quickly compute effective resistance estimates, which were then used to construct a spectral sparsifier. 

Since a ${1}/{\epsilon^2}$ term appears in both of these natural ways to compactly represent effective resistances, 
an obvious question is whether this is \emph{necessary}.  
For the stronger requirement of spectral sparsification, 
we know the answer is yes -- every spectral sparsifier of the complete graph requires $\Omega(n/\epsilon^2)$ edges~\cite[Section 4]{BSS12}
(see also \cite{AKW14}).
However, it is currently unknown whether such a bound holds also for resistance sparsifiers, and the starting point of our work is the observation (based on~\cite{vLRH14}) that for the complete graph, 
every $O({1}/{\epsilon})$-regular expander is a $(1+\epsilon)$-resistance sparsifier, despite not being a $(1+\epsilon)$-spectral sparsifier!  
We thus put forward the following conjecture.

\begin{conjecture} \label{conj:resistance}
Every graph admits a $(1+\epsilon)$-resistance sparsifier 
with $\tO(n/\epsilon)$ edges.
\end{conjecture}

We make the first step in this direction by proving the special case of 
dense regular expanders (which directly generalize the complete graph).  
Even this very special case turns out to be nontrivial, and in fact leads us 
to another beautiful problem which is interesting in its own right.  

\begin{question} \label{q:expander}
Does every dense regular expander contain a sparse regular expander
as a subgraph?
\end{question}

Our positive answer to this question (for a certain definition of expanders)
forms the bulk of our technical work (Sections \ref{sec:subgraphs} and~\ref{sec:krv_extended_proof}),
and is then used to find good resistance sparsifiers 
for dense regular expanders (Section~\ref{sec:resistance}).

\subsection{Results and Techniques} \label{sec:results}

Throughout, we consider undirected graphs, and they are unweighted unless stated otherwise.  
In a weighted graph, i.e., when edges have nonnegative weights,
the \emph{weighted degree} of a vertex is the sum of weights on incident edges, and the graph is considered regular if all of its weighted degrees are equal. Typically, a sparsifying subgraph must be weighted even when the host graph is unweighted, in order to exhibit comparable parameters with far fewer edges.

Before we can state our results we first need to recall some basic definitions from spectral graph theory.  Given a weighted graph $G$, let $D$ be the diagonal $n \times n$ matrix of weighted degrees, and let $A$ be the weighted adjacency matrix.  
The \emph{Laplacian} of $G$ is defined as $L\eqdef D-A$, and the \emph{normalized Laplacian} is the matrix $\hat L\eqdef D^{-1/2} L D^{-1/2}$.

\begin{definition}[Effective Resistance] \label{def:resistance_distance}
Let $G(V,E,w)$ be a weighted graph, and let $P$ the Moore-Penrose pseudo-inverse of its Laplacian matrix. 
The \emph{effective resistance} (also called \emph{resistance distance})
between two nodes $u,v\in V$ is 
\begin{equation*}
  R_G(u,v) \eqdef (e_u-e_v)^TP(e_u-e_v) , 
\end{equation*}
where $e_u$ and $e_v$ denote the standard basis vectors in $\R^V$ that correspond to $u$ and $v$ respectively.
\end{definition}
When the graph $G$ is clear from context we will omit it and write $R(u,v)$.  
We can now define the main objects that we study.

\begin{definition}[Resistance Sparsifier]
Let $G(V,E,w)$ be a weighted graph, and let $\epsilon\in(0,1)$. 
A \emph{$(1+\epsilon)$-resistance sparsifier} for $G$ 
is a subgraph $H(V,E',w')$ with reweighted edges such that
$  (1-\epsilon) R_H(u,v) \leq R_G(u,v) \leq (1 + \epsilon)R_H(u,v)$,
for all $u,v \in V$.
\end{definition}

It will turn out that in order to understand resistance sparsifiers, we need to use expansion properties. 

\begin{definition}[Graph Expansion]
The \emph{edge-expansion} (also known as the \emph{Cheeger constant}) 
of a weighted graph $G(V,E,w)$ is 
\[ 
  \phi(G) \eqdef 
  \min \left\{ \frac{w(S, \bar S)}{|S|} : 
    S\subset V,\ 0<|S|\leq|V|/2 \right\} ,
\]
where $w(S, \bar S)$ denotes the total weight of edges with exactly one 
endpoint in $S\subset V$.  
The \emph{spectral expansion} of $G$, denoted $\lambda_2(G)$,
is the second-smallest eigenvalue of the graph's normalized Laplacian.
\end{definition}

Our main result is the following.   Throughout this paper, ``efficiently'' means in randomized polynomial time.

\begin{theorem}\label{thm:main_resistance}
Fix $\beta,\gamma>0$, let $n$ be sufficiently large, and ${1}/{n^{0.99}}<\epsilon<1$.
Every $D$-regular graph $G$ on $n$ nodes with $D\geq\beta n$ and $\phi(G)\geq\gamma D$ contains (as a subgraph) a $(1+\epsilon)$-resistance sparsifier with at most $\epsilon^{-1}n(\log n)^{O(1/\beta\gamma^2)}$ edges, and it can be found efficiently.
\end{theorem}

While dense regular expanders may seem like a simple case, even this special case requires significant technical work.  The most obvious idea, of sparsifying through random sampling, does not work ---
selecting each edge of $G$ uniformly at random with probability $\tO(1/(D\epsilon))$ (the right probability for achieving a subgraph with $\tO(n/\epsilon)$ edges) need not yield a $(1+\epsilon)$-resistance sparsifier. 
Intuitively, this is because the variance of independent random sampling is too large (see Theorem~\ref{thm:von_luxburg} for the precise effect),
and the easiest setting to see this is the case of sparsifying the complete graph.  
If we sparsify through independent random sampling, then to get a $(1+\epsilon)$-resistance sparsifier requires picking each edge independently with probability at least $1/(\epsilon^2 n)$, 
and we end up with $n/\epsilon^2$ edges. 
To beat this, we need to use correlated sampling.  More specifically, it turns out that a random $O(1/\epsilon)$-regular graph is a $(1+\epsilon)$-resistance sparsifier of the complete graph, despite not being a $(1+\epsilon)$-spectral sparsifier.  So instead of sampling edges independently (the natural approach, and in fact the approach used to construct spectral sparsifiers by Spielman and Srivastava~\cite{SS11}), we need to sample a random regular graph.  

In order to prove Theorem~\ref{thm:main_resistance}, we actually need to generalize this approach beyond the complete graph.  But what is the natural generalization of a random regular graph when the graph we start with is not the complete graph?  
It turns out that what we need is an expander, which is sparse but maintains regularity of its degrees.  
This motivates our main structural result, that every dense regular expander contains a sparse regular expander (as a subgraph). 
This can be seen as a type of sparsification result that retains regularity.

\begin{theorem}\label{thm:main}
Fix $\beta,\gamma>0$ and let $n$ be sufficiently large. 
Every $D$-regular graph $G$ on $n$ nodes with $D\geq\beta n$ and $\phi(G)\geq\gamma D$ contains a weighted $d$-regular subgraph $H$ with $d=(\log n)^{O(1/\beta\gamma^2)}$ and $\phi(H)\geq\frac{1}{3}$. All edge weights in $H$ are in $\{1,2\}$, and $H$ can be found efficiently.
\end{theorem}

To prove this theorem, we analyze a modified version of the cut-matching game of Khandekar, Rao, and Vazirani~\cite{KRV09}.  
This game has been used in the past to construct expander graphs, but in order to use it for Theorem~\ref{thm:main} we need to generalize beyond matchings, and also show how to turn the graphs it creates (which are not necessarily subgraphs of $G$) into subgraphs of $G$.  

The expansion requirement for $G$ in \Cref{thm:main} is equivalent to $\lambda_2(G)=\Omega(1)$, 
when $\beta$ and $\gamma$ are viewed as absolute constants. 
We note that $H$ is a much weaker expander, satisfying only $\lambda_2(H)=\Omega(1/\mathrm{polylog}(n))$, but this is nonetheless sufficient for \Cref{thm:main_resistance}. Also, $H$ is regular in weighted degrees. For completeness we give a variant of \Cref{thm:main} that achieves an unweighted $H$ by requiring stronger expansion from $G$, but this is not necessary for our application to resistance sparsifiers, which anyway involves reweighting the edges.

\begin{theorem}\label{thm:unweighted}
For every $\beta>0$ there is $0<\gamma<1$ such the following holds for sufficiently large $n$. Every $D$-regular graph $G$ on $n$ nodes with $D\geq\beta n$ and $\phi(G)\geq\gamma D$ contains an (unweighted) $d$-regular subgraph $H$ with $d=(\log n)^{O(1/\beta\gamma)}$ and $\phi(H)\geq\frac{1}{3}$, and it can be found efficiently.
\end{theorem}

The algorithm underlying Theorems~\ref{thm:main_resistance},~\ref{thm:main} and~\ref{thm:unweighted} turns out to be quite straightforward: decompose the host graph into disjoint perfect matchings or Hamiltonian cycles (which are ``atomic'' regular components), and subsample a random subset of them of size $d$ to form the target subgraph. However, since the decomposition leads to large dependencies between inclusion of different edges in the subgraph, it is unclear how to approach this algorithm with direct probabilistic analysis. Instead, our analysis uses the adaptive framework of~\cite{KRV09} to quantify the effect of gradually adding random matching/cycles from the decomposition to the subgraph.

\subsection{Related Work} \label{sec:related}

The line of work most directly related to resistance sparsifiers is the construction of spectral sparsifiers.  This was initiated by Spielman and Teng~\cite{ST04a}, and was later pushed to its limits by Spielman and Teng~\cite{ST11}, Spielman and Srivastava~\cite{SS11}, and Batson, Spielman, and Srivastava~\cite{BSS12}, who finally proved that every graph has a $(1+\epsilon)$-spectral sparsifier with $O(n/\epsilon^2)$ edges and that this bound is tight (see also \cite{AKW14}).  

The approach by Spielman and Srivastava~\cite{SS11} is particularly closely related to our work.  They construct almost-optimal spectral sparsifiers (a logarithmic factor worse than~\cite{BSS12}) by sampling each edge independently with probability proportional to the effective resistance between the endpoints.  This method naturally leads us to try the same thing for resistance sparsification, but as discussed, independent random sampling (even based on the effective resistances) cannot give improved resistance sparsifiers.  Interestingly, in order to make their algorithm extremely efficient they needed a way to estimate effective resistances very quickly, so along the way they showed how to create a  sketch of size $O(n\log n/\epsilon^2)$ from which every resistance distance can be read off in $O(\log n)$ time (essentially through an $\ell_2$-squared embedding and a Johnson-Lindenstrauss dimension reduction).

\section{Sparse Regular Expanding Subgraphs} \label{sec:subgraphs}

In this section we prove \Cref{thm:main}, building towards it in stages.
Our starting point is the Cut-Matching game of 
Khandekar, Rao and Vazirani (KRV) \cite{KRV09},
which is a framework to constructing sparse expanders 
by iteratively adding perfect matchings across adaptively chosen bisections 
of the vertex set. 
The resulting graph $H$ is regular, as it is the union of perfect matchings,
and if the matchings are contained in the input graph $G$ 
then $H$ is furthermore a subgraph of $G$, as desired. 
In \Cref{sec:cut-matching}, we employ this approach
to prove \Cref{thm:main} in the case $D/n = \frac{3}{4}+\Omega(1)$.

To handle smaller $D$, we observe that the perfect matchings in the KRV game can be replaced with a more general structure that we call a \emph{weave},
defined as a set of edges where for every vertex at least one incident edge 
crosses the given bisection. 
To ensure that $H$ is regular (all vertices have the same degree), 
we would like the weaves to be regular. 
We thus decompose the input graph to disjoint regular elements -- 
either perfect matchings or Hamiltonian cycles -- 
and use them as building blocks to construct regular weaves. 
Leveraging the fact that for some bisections, 
$G$ contains no perfect matching but does contain a weave, 
we use this extension in \Cref{sec:Cut-Weave}
to handle the case $D/n = \frac{1}{2}+\Omega(1)$.

Finally, for the general case $D/n=\Omega(1)$, 
we need to handle a graph $G$ that contains no weave on some bisections. 
The main portion of our proof constructs a weave that is not contained in $G$,
but rather embeds in $G$ with small (polylogarithmic) congestion. 
Repeating this step sufficiently many times as required by the KRV game,
yields a subgraph $H$ as desired.

\paragraph*{Notation and terminology.} 
For a regular graph $G$, we denote $\mathrm{deg}(G)$ the degree of each vertex. We say that a graph $H$ is an \emph{edge-expander} if $\phi(H)>\frac{1}{3}$. A \emph{bisection} of a vertex set of size $n$ is a partition $(S,\bar S)$ with equal sizes $\frac{1}{2}n$ if $n$ is even, or with sizes $\lfloor\frac{1}{2}n\rfloor$ and $\lceil\frac{1}{2}n\rceil$ if $n$ is odd.

\subsection{The Cut-Matching Game} \label{sec:cut-matching}
Khandekar, Rao and Vazirani \cite{KRV09} described the following game 
between two players. 
Start with an empty graph (no edges) $H$ on a vertex set of even size $n$. 
In each round, the \emph{cut player} chooses a bisection, 
and the \emph{matching player} answers with a perfect matching across the bisection. 
The game ends when $H$ is an edge-expander.
Informally, the goal of the cut player is to reach this as soon as possible,
and that of the matching player is to delay the game's ending. 

\begin{theorem}[\cite{KRV09,KKOV07}]\label{thm:krv}
The cut player has an efficiently computable strategy that wins
(i.e., is guaranteed to end the game) within $O(\log^2n)$ rounds, 
and a non-efficient strategy that wins within $O(\log n)$ rounds.
\end{theorem}

The following result illustrates the use of the KRV framework in our setting.
\begin{theorem}\label{thm:d34n}
Let $\delta>0$ and let $n$ be even and sufficiently large ($n\ge n_0(\delta)$).
Then every $n$-vertex graph $G(V,E)$ with minimum degree $D\geq(\frac{3}{4}+\delta)n$ contains an edge-expander $H$ 
that is $d$-regular for $d=O(\log n)$, 
and also an efficiently computable edge-expander $H'$ 
that is a $d'$-regular for $d'=O(\log^2n)$.
\end{theorem}

\begin{proof}
Apply the Cut-Matching game on $V$ with the following player strategies. 
For the cut player, execute the efficient strategy from \Cref{thm:krv} that wins within $O(\log^2n)$ rounds. 
For the matching player, given a bisection $(S,\bar S)$, consider the bipartite subgraph $G[S,\bar S]$ of $G$ induced by $(S,\bar S)$. 
Each vertex in $S$ has in $G$ at least $D\geq\frac{3}{4}n$ neighbors, 
but at most $\frac{1}{2}n-1$ of them are in $S$, 
and the rest must be in $\bar S$, 
which implies that $G[S,\bar S]$ has minimum degree $\geq\frac{1}{4}n$. 
Hence, as a simple consequence of Hall's theorem (see \Cref{prp:hall_consequence}), it contains a perfect matching that can be efficiently found. 
The matching player returns this matching as his answer.
We then remove this matching from $G$ before proceeding to the next round, 
to ensure that different iterations find disjoint matchings. 
The slackness parameter $\delta$ (and $n$ being sufficiently large) 
ensure that the minimum degree of $G$ does not fall below $\frac{3}{4}n$ during the $O(\log^2n)$ iterations, so the above argument holds in all rounds.

The game ends with an edge-expander $H'$ which is a disjoint union of $d'=O(\log^2n)$ perfect matchings contained in $G$, 
and hence is a $d'$-regular subgraph of $G$, as required. 
To obtain the graph $H$, apply the same reasoning but using the non-efficient strategy from \Cref{thm:krv} that wins within $O(\log n)$ rounds.
\end{proof}

\subsection{The Cut-Weave Game} \label{sec:Cut-Weave}

For values of $D$ below $\frac{3}{4}n$, we can no longer guarantee that every bisection in $G$ admits a perfect matching. However, we observe that one can allow the matching player a wider range of strategies while retaining the ability of the cut player to win within a small number of rounds.

\begin{definition}[weave]
Given a bisection $(S,\bar S)$ of a vertex set $V$, a \emph{weave} on $(S,\bar S)$ is a subgraph in which every node has an incident edge crossing $(S,\bar S)$.
\end{definition}

\begin{definition}[Cut-Weave Game]
The \emph{Cut-Weave game} with parameter $r$ is the following game of two players. Start with a graph $H$ on a vertex set of size $n$ and no edges. In each round, the \emph{cut player} chooses a bisection of the vertex set, and the \emph{weave player} answers with an $r$-regular weave on the bisection. The edges of the weave are added to $H$.
\end{definition}
Note that the $r=1$ case is the original Cut-Matching game (when $n$ is even). 
The following theorem is an extension of \Cref{thm:krv}. For clarity of presentation, its proof is deferred to  \Cref{sec:krv_extended_proof}. 

\begin{theorem}\label{thm:krv_extended}
In the Cut-Weave game with parameter $r$, the cut player has an efficient strategy that wins within $O(r\log^2n)$ rounds, and furthermore ensures $\phi(H)\geq\frac{1}{2}r$.
\end{theorem}
In order to construct regular weaves, we employ a decomposition of $G$ into disjoint Hamiltonian cycles. The following theorem was proven by Perkovic and Reed \cite{PR97}, and recently extended by Csaba, K\"{u}hn, Lo, Osthus and Treglown~\cite{CKLOT14}.

\begin{theorem}\label{thm:Hamiltonian_decomposition}
Let $\delta>0$. Every $D$-regular graph $G$ on $n$ nodes with $D\geq(\frac{1}{2}+\delta)n$, admits a decomposition of its edges into $\lfloor\frac{1}{2}D\rfloor$ Hamiltonian cycles and possibly one perfect matching (if $D$ is odd). Furthermore, the decomposition can be found efficiently.
\end{theorem}

Now we can use the Cut-Weave framework to make another step towards \Cref{thm:main}.
\begin{theorem}\label{thm:d12n}
Let $\delta>0$ and let $n$ be sufficiently large. 
Then every $n$-vertex graph $G(V,E)$ with minimum degree $D\geq(\frac{1}{2}+\delta)n$ contains a $d$-regular edge-expander $H$ with $d=O(\log^3n)$, which furthermore can be efficiently found.
\end{theorem}

\begin{proof}
We simulate the Cut-Weave game with $r=16\delta^{-1}\log n$. The proof is the same as \Cref{thm:d34n}, only instead of a perfect matching we need to construct an $r$-regular weave across a given bisection $(S,\bar S)$. We apply \Cref{thm:Hamiltonian_decomposition} to obtain a Hamiltonian decomposition of $G$. For simplicity, if $D$ is odd we discard the one perfect matching from \Cref{thm:Hamiltonian_decomposition}. Let $\mathcal C$ be the collection of Hamiltonian cycles in the decomposition.

Suppose w.l.o.g.~$|S|=\lceil\frac{1}{2}n\rceil$. Every $v\in S$ has at most $|S|-1\leq\frac{1}{2}n$ neighbors in $S$, and hence at least $\delta n$ incident edges crossing to $\bar S$. We set up a Set-Cover instance of the cycles $\mathcal C$ against the nodes in $S$, where a node $v$ is considered covered by a cycle $C$ is $v$ has an incident edge crossing to $\bar S$, that belongs to $C$. This is a dense instance: since each cycle visits $v$ only twice, $v$ can be covered by $\frac{1}{2}\delta n$ cycles. Therefore, $4\delta^{-1}\log n$ randomly chosen cycles form a cover with high probability (see \Cref{prp:dense_set_cover} for details). We then repeat the same procedure to cover the nodes on side $\bar S$. The result is a collection of $8\delta^{-1}\log n=\frac{1}{2}r$ disjoint Hamiltonian cycles, whose union forms an $r$-regular weave on $(S,\bar S)$, which we return as the answer of the weave player. Applying \Cref{thm:krv_extended} with $r=O(\log n)$ concludes the proof of \Cref{thm:d12n}.
\end{proof}

Observe that in the proof of Theorem~\ref{thm:d12n}, the weave player is in fact oblivious to the queries of the cut player: all she does is sample random cycles from $\mathcal C$, and the output subgraph $H$ is the union of those cycles. Therefore, in order to construct $H$, it is sufficient to decompose $G$ into disjoint Hamiltonian cycles, and choose a random subset of size $O(\log^3n)$ of them. There is no need to actually simulate the cut player, and in particular, the proof does not require her strategy (from Theorem~\ref{thm:krv_extended}) to be efficient.

\subsection{Reduction to Double Cover}\label{sec:double_cover}
We now begin to address the full range of parameters stated in \Cref{thm:main}. In this range there is no Hamiltonian decomposition theorem (or a result of similar flavor) that we are aware of, so we replace it with a basic argument which incurs edge weights $w:V\times V\rightarrow\{0,1,2\}$ in the target subgraph $H$, as well as a loss in its degree.

Given the input graph $G(V,E)$, we construct its \emph{double cover}, which is the bipartite graph $G''(V'',E'')$ defined by $V''=V\times\{0,1\}$ and $E'' = \{((v,0)(u,1)) : vu\in E \}$. It is easily seen that if $G$ is $D$-regular then so is $G''$, and since $|V''|=2|V|$ we have $D\geq\frac{1}{2}\beta|V''|$. It also well known that $\lambda_2(G)=\lambda_2(G'')$, and therefore by the discrete Cheeger inequalities,
\begin{equation*}
  \phi(G'') \geq 
  \tfrac{1}{2}\lambda_2(G'')D = 
  \tfrac{1}{2}\lambda_2(G)D \geq 
  \tfrac{1}{2}\gamma^2(G)D.
\end{equation*}

$G''$ satisfies the requirements of \Cref{thm:main} with $\beta''=\frac{1}{2}\beta$ and $\gamma''=\tfrac{1}{2}\gamma^2$. Suppose we find in $G''$ a $d$-regular edge-expander $H''$ with $d = (\log n)^{O(1/\beta''\gamma'')} = (\log n)^{O(1/\beta\gamma^2)}$. We carry it over to a subgraph $H$ of $G$, by including each edge $uv\in E$ in $H$ with weight $\left|\{(v,0)(u,1),(u,0)(v,1)\} \cap  E(H'')\right|$, where $E(H'')$ denotes the edge set of $H''$. Each edge then appears in $H$ with weight either $1$ or $2$ (or $0$, which means it is not present in $H$). It can be easily checked that $H$ is $d$-regular in weighted degrees, and $\phi(H)\geq\frac{1}{2}\phi(H'')$. Therefore $H$ is a suitable target subgraph for \Cref{thm:main}.

The above reduction allows us to restrict our attention to regular bipartite graphs $G$, but on the other hand we are forced to look for a subgraph $H$ which is unweighted and $d$-regular with $d=(\log n)^{O(1/\beta\gamma)}$ (which is tighter than stated in \Cref{thm:main}). We take this approach in the remainder of the proof. The gain is that such $G$ admits a decomposition into disjoint perfect matchings, which can be efficiently found, as a direct consequence of Hall's theorem. We will use this fact where we have previously used  
\Cref{thm:Hamiltonian_decomposition}.

\subsection{Constructing an Embedded Weave}
We now get to the main technical part of the proof. Given a bisection $(S,\bar S)$ queried by the cut player, we need to construct an $r$-regular weave on the bisection, where this time we choose $r=(\log n)^{O(1/\beta\gamma)}$. Unlike the proof of \Cref{thm:d12n}, we cannot hope to find a weave which is a subgraph of $G$, since if $D<\frac{1}{2}n$, any bisection in which one side contains some vertex and all its neighbors would not admit a weave in $G$. Instead, we aim for a weave which embeds into $G$ with polylogarithmic congestion.

We will use two types of graph operations: The \emph{union} of two graphs on the same vertex set $V$ is obtained by simply taking the set union of their edge sets, whereas the \emph{sum} of the two graphs is given by keeping parallel edges if they appear in both graphs. We now construct the weave in 4 steps.

\paragraph*{Step 1.}
Fix $\mu=\frac{\beta\gamma^2}{4}$. We partition the entire vertex set $V$ into subsets $S_0,S_1,\ldots,S_t$ by the following process:
\begin{enumerate}
  \item Set $S_0 \leftarrow \bar S$ and $T \leftarrow S$.
  \item While $T\neq\emptyset$, take $S_i\subseteq T$ to be the subset of nodes with at least $\mu D$ neighbors in $S_{i-1}$, and set $T \rightarrow T\setminus S_i$.
\end{enumerate}
\begin{mylemma}
The process terminates after $t\leq\frac{2}{\beta\gamma}$ iterations.
\end{mylemma}
\begin{proof}
Consider an iteration $i\leq\frac{2}{\beta\gamma}$ that ends with $T\neq\emptyset$. Denote $\bar T=V\setminus T=\cup_{j=0}^iS_j$. By the hypothesis $\phi(G)\geq\gamma D$ we have at least $\gamma D|T|$ edges crossing from $T$ to $\bar T$, so by averaging over the nodes in $T$, there is $v\in T$ with $\gamma D$ neighbors in $\bar T$. For every $j<i$, $v$ must have less than $\mu D$ neighbors in $S_j$, or it would already belong to $S_{j+1}\subseteq\bar T$. Summing over $j=0,\ldots,i-1$, we see that $v$ has less than $i\mu D\leq\frac{1}{2}\gamma D$ neighbors in $\bar T\setminus S_i$, so at least $\frac{1}{2}\gamma D$ neighbors in $S_i$. This implies $|S_i|\geq\frac{1}{2}\gamma D$. We have shown that each of the first $\frac{2}{\beta\gamma}$ iterations either terminates the process or removes $\frac{1}{2}\gamma D\geq\frac{1}{2}\gamma\beta n$ nodes from $T$, so after $\frac{2}{\beta\gamma}$ iterations we must have $T=\emptyset$.
\end{proof}

\paragraph*{Step 2.}
By \Cref{sec:double_cover} we have a decomposition of all the edges in $G$ into a collection $\mathcal M$ of $D$ disjoint perfect matchings. For every $i=1,\ldots,t$, we now cover the nodes in $S_i$ with perfect matchings, similar to the proof of \Cref{thm:d12n}. A node $v\in S_i$ is considered covered by a matching if $v$ has an incident edge with the other endpoint in $S_{i-1}$, and that edge lies on the matching. Since $v$ has $\mu D$ incident edges crossing to $S_{i-1}$, and each matching touches $v$ with at most one edge, we have $\mu D$ matchings that can cover $v$. Therefore $k=\frac{1}{\mu}\log n$ randomly chosen matchings from $\mathcal M$ form a cover of $S_i$ (see \Cref{prp:dense_set_cover}), which we denote as $K_i$. Thus, for each $i$ we have a subgraph $K_i$ which is $k$-regular, such that each node in $S_i$ has an incident edge in $K_i$ with the other endpoint in $S_{i-1}$. Denote henceforth
\[ K = \cup_{i=1}^t K_i. \]
Note that $K$ is a regular subgraph of $G$, since it is a union of disjoint perfect matchings from $\mathcal M$,
and $\mathrm{deg}(K)\leq kt$.

\paragraph*{Step 3.}
In this step we construct a graph $K^*$ from the subgraph $K$. As discussed, $K^*$ will not be a subgraph of $G$ but will embed into it with reasonable congestion. 
Let us formally define the notion of graph embedding that we will be using.
\begin{definition}[Graph embedding with congestion]
Let $G(V,E)$ and $G'(V,E')$ be graphs on the same vertex set. Denote by $\mathcal{P}_G$ the set of simple paths in $G$. An \emph{embedding} of $G'$ into $G$ is a map $f:E'\rightarrow\mathcal{P}_G$ such that every edge in $G'$ is mapped to a path in $G$ with the same endpoints.

The \emph{congestion} of $f$ on an edge $e\in E$ is $\mathrm{cng}_f(e):=|e'\in E':e\in f(e')|$. The congestion of $f$ is $\mathrm{cng}(f):=\max_{e\in E}\mathrm{cng}_f(e)$. We say that $G'$ embeds into $G$ with congestion $c$ if there is an embedding $f$ with $\mathrm{cng}(f)=c$.
\end{definition}
The following claim is a simple observation and we omit its proof.
\begin{claim}
If $G'$ embeds into $G$ with congestion $c$, then $\phi(G)\geq\frac{1}{c}\phi(G')$.
\end{claim}

We generate $K^*$ with the following inductive construction.

\begin{mylemma}\label{clm:main_construction}
Let $\rho_0=c_0=0$. We can efficiently construct subgraphs $K_1^*,\ldots,K_t^*$ (which may have parallel edges and self-loops), such that for every $i=1,\ldots,t$,
\begin{enumerate}
  \item $K_i^*$ is $\rho_i$-regular, where $\rho_i=k(1+\rho_{i-1})$.
  \item $K_i^*$ embeds into $K$ with congestion $c_i$, where $c_i=1+kc_{i-1}$.
  \item Every $v\in S_i$ has an incident edge in $K_i^*$ with the other endpoint in $S_0$.
\end{enumerate}
\end{mylemma}
\begin{proof}
We go by induction on $i$. For the base case $i=1$ we simply set $K_1^*=K_1$. The claim holds as we recall that
\begin{enumerate}
  \item $K_1$ is $k$-regular.
  \item $K_1$ is a subgraph of $K$, hence it embeds into $K$ with congestion $1=1+kc_0$.
  \item By Step 2, every $v\in S_1$ has an incident edge in $K_1$ crossing to $S_0$.
\end{enumerate}

We turn to the inductive step $i>1$. Start with a graph $K'$ which is a fresh copy of $K_{i-1}^*$, with each edge duplicated into $k$ parallel edges. By induction, $K'$ is $(k\rho_{i-1})$-regular. Now sum $K_i$ into $K'$; recall this means keeping parallel edges instead of unifying them. Since $K_i$ is $k$-regular, $K'$ is $\rho_i$-regular.

Let $v\in S_i$. By Step 2, there is an edge $vw\in K_i$ such that $w\in S_{i-1}$. By induction, there is an edge $wu\in K_{i-1}^*$ such that $u\in S_0$. Note that both edges $vw$ and $wu$ are present in $K'$. Perform the following crossing operation on $K'$:
Remove the edges $vw$ and $wu$, and add an edge $vu$ and a self-loop on $w$.

Perform this on every $v\in S_i$. The resulting graph is $K_i^*$. We need to show that it is well defined in the following sense: we might be using the same edge $wu$ for several $v$'s, and we need to make sure each $wu$ appears sufficiently many times, to be removed in all the crossing operations in which it is needed. Indeed, we recall that $K_i$ is the union of $k$ disjoint perfect matchings, and therefore each $w\in S_{i-1}$ has at most $k$ edges in $K_i$ incoming from $S_i$. Since $K'$ contains $k$ copies of each edge $wu$, we have enough copies to be removed in all necessary crossing operations.

Lastly we show that $K_i^*$ satisfies all the required properties.
\begin{enumerate}
  \item Since $K'$ was $\rho_i$-regular, and the switching operations do not effect vertex degrees, we see that $K_i^*$ is $\rho_i$-regular.
  \item Each edge $vu$ in $K_i^*$ which is not original from $K'$, corresponds to a path (of length 2) in $K'$ that was removed upon adding that edge; hence $K_i^*$ embeds into $K'$ with congestion $1$. $K'$ is the sum of $K_i$, which is a subgraph of $K$, and $k$ copies of $K_{i-1}^*$, which by induction embeds into $K$ with congestion $c_{i-1}$. Hence $K'$ embeds into $K$ with congestion $1+kc_{i-1}=c_i$. Therefore, $K_i^*$ embeds into $K$ with congestion $c_i$.
  \item For every $v\in S_i$, we added to $K_i^*$ an edge $vu$ such that $u\in S_0$.
\end{enumerate}
\end{proof}

We now take
$K^* = \sum_{i=1}^tK_i^*$. By \Cref{clm:main_construction}, $K^*$ is $(\sum_{i=1}^t\rho_i)$-regular, embeds into $K$ with congestion $\sum_{i=1}^tc_i$, and every $v\in S$ has an incident edge $vu\in K^*$ such that $u\in \bar S$. (To see why the latter point holds, recall that we put $\bar S=S_0$.)

\paragraph*{Step 4.}
In this final step we repeat Steps 1--3, only with the roles of $S$ and $\bar S$ interchanged. This results in a subgraph $\bar K$ of $G$ which is $kt$-regular, and a graph $\bar K^*$ which is $(\sum_{i=1}^t\rho_i)$-regular, embeds into $\bar K$ with congestion $\sum_{i=1}^tc_i$, and every $v\in\bar S$ has an incident edge $vu\in \bar K^*$ such that $u\in S$.

Our final weave is $K^*+\bar K^*$. By the above it is clearly a weave, and moreover it is $r$-regular and embeds into $K\cup\bar K$ (and hence into $G$, which contains $K\cup\bar K$) with congestion $c$, where
$r = 2\sum_{i=1}^t\rho_i$ and $c = 2\sum_{i=1}^tc_i$.
By inspecting the recurrence formulas from \Cref{clm:main_construction}, in which $\rho_i$ and $c_i$ were defined, we can bound
$\rho_i,c_i \leq (2k)^i \leq (2k)^t$ for every $i$,
and hence
$r,c \leq 2t(2k)^t$.
Recalling that $t\leq\frac{2}{\beta\gamma}+1$ and $k=\frac{1}{\mu}\log n=O(\log n)$, we find 
$r,c\leq(\log n)^{O(1/\beta\gamma)}$.

\subsection{Completing the Proof of \Cref{thm:main}}
We play the Cut-Weave game for $L$ rounds, where $L=O(r\log^2n)$ is the number of rounds required by the efficient strategy in \Cref{thm:krv_extended}. For each round $\ell=1,\ldots,L$, we constructed above an $r$-regular weave $W^*_\ell=K^*+\bar K^*$, that embeds into a subgraph $W_\ell=K\cup\bar K$ of $G$ with congestion $c$. Let
$H=\cup_{\ell=1}^L W_\ell$ and $H^*=\sum_{1=\ell}^L W^*_\ell$.  Then $H$ is a union of disjoint perfect matchings from $\mathcal M$, and hence regular. Moreover $\mathrm{deg}(H)\leq 2ktL$, since $H$ is the union of $L$ subgraphs $\{W_\ell\}_{\ell=1}^L$, where each $W_\ell$ is a union $W_\ell$ of two $kt$-regular graphs $K,\bar K$.

Now consider $H^*$. Since each $W^*_\ell$ embeds into $W_\ell$ with congestion $c$, we see that $H^*$ embeds into $H$ with congestion (at most) $cL$. By \Cref{thm:krv_extended} we have $\phi(H^*)\geq\frac{1}{2}r$, and this now implies $\phi(H)\geq\frac{r}{2cL}$.

Recalling the parameters:
\[ t = O(1) \;\; ; \;\; k = O(\log n) \;\; ; \;\; r,c=O(\log^{O(1/\beta\gamma)}n) \;\; ; \;\;
L=O(r\log^2n), \]
we see that $H$ is a $d$-regular subgraph of $d=(\log n)^{O(1/\beta\gamma)}$ and $\phi(H)\geq1/(\log n)^{O(1/\beta\gamma)}$. We can now repeat this Cut-Weave game $(\log n)^{O(1/\beta\gamma)}$ disjoint times, because if each time we remove the graph $H$ we have found, we decrease the degree $D=\beta n$ of each node by only $\polylog{n}$. By repeating the game this many times and taking the union of the disjoint resulting subgraphs, we find a regular subgraph $H$ of $G$ with $\mathrm{deg}(H)=(\log n)^{O(1/\beta\gamma)}$ and $\phi(H)\geq1$. Lastly recall that unfolding the reduction from \Cref{sec:double_cover} puts on $H$ edge weight in $\{1,2\}$, and weakens the degree bound to $\mathrm{deg}(H)=(\log n)^{O(1/\beta\gamma^2)}$. This completes the proof of \Cref{thm:main}.

Regarding the algorithm to construct $H$, the observation made after Theorem~\ref{thm:d12n} applies here as well. The weave player's strategy is oblivious to the queries of the cut player, since she just samples random matchings from $\mathcal M$ to form $H$. The cut player strategy does not actually need to be simulated, nor the graphs $K^*$ need to actually be constructed. The algorithm to construct $H$ then amounts to the following: Construct the double cover graph $G"$ of $G$; decompose $G"$ into disjoint perfect matchings; choose a random subset of $(\log n)^{O(1/\beta\gamma^2)}$ of them to form a subgraph $H"$ of $G"$; and unfold the double cover construction to obtain the final subgraph $H$ from $H"$. 

\subsection{Proof of \Cref{thm:unweighted}}
The theorem follows from replacing the reduction to the double cover in \Cref{sec:double_cover} by a Hamiltonian decomposition result that holds for this stronger expansion requirement, due to K\"{u}hn and Osthus \cite[Theorem 1.11]{KO14}. The trade-off between $\beta$ and $\gamma$ is inherited from their theorem (in which it is unspecified). Circumventing \Cref{sec:double_cover} also improves the dependence of $d$ on $\gamma$. The proof of \Cref{thm:unweighted} is otherwise identical to the proof of \Cref{thm:main}.

\section{Proof of the Cut-Weave Theorem}\label{sec:krv_extended_proof}
Recall the setting of the Cut-Weave game with parameter $r$: The game starts with a graph $G_0$ on $n$ vertices and without edges. In each round $t=1,2,\ldots$, the weave player queries a bisection of the vertex set, and the weave player answers with an $r$-regular weave $H_t$ on that bisection. The weave is then unified into the graph, putting $G_t=G_{t-1}\cup H_t$.

We now prove \Cref{thm:krv_extended} by an adaptation of the analysis from \cite{KRV09}. The main change is in \Cref{lmm:krv_potential_reduction}.

For each step $t$, let $M_t$ be the matrix describing one step of the natural lazy random walk on $H_t$: W.p.~$\frac{1}{2}$ stay in the current vertex, and with probability $\frac{1}{2r}$ move to a neighbor. The cut player strategy is as follows:
\begin{itemize}
  \item Choose a random unit vector $z\perp\mathbf1$ in $\R^n$.
  \item Compute $u=M_tM_{t-1}\ldots M_1z$.
  \item Output the bisection $(S,\ldots S)$ where $S$ is the $\lfloor n/2 \rfloor$ vertices with smallest values in $u$.
\end{itemize}

Let us analyze the game with this strategy. In the graph $G_t$ (which equals $\cup_{t'=1}^tH_{t'}$), we consider the following $t$-steps random walk: Take one (lazy) step on $H_1$, then on $H_2$, and so on until $H_t$. In other words, the walk is given by applying sequentially $M_1$, then $M_2$, and so on.

Let $P_{ij}(t)$ denote the probability to go from node $j$ to node $i$ within $t$ steps. Let $P_i$ denote the vector $(P_{i1},P_{i2},\ldots,P_{ji})$. We use the following potential function:
\[ \Psi(t) = \sum_{i,j\in V}(P_{ij}-1/n)^2 = \sum_{i=1}^n\norm{P_i-\mathbf1/n}_2^2. \]

\begin{mylemma}
For every $t$ and every $i\in V$, we have $\sum_{j\in V}P_{ij}(t)=1$.
\end{mylemma}
\begin{proof}
By induction on $t$: It holds initially, and in each step $t$, vertex $i$ trades exactly half of its total present probability with its neighbors in $H_t$. (Note that this relies on the fact that $H_t$ is regular.)
\end{proof}

\begin{mylemma}\label{lmm:krv_potential_to_expansion}
If $\Psi(t)<1/4n^2$ then $G=G_t$ has edge-expansion at least $\frac{1}{2}r$.
\end{mylemma}
\begin{proof}
If $\Psi(t)<1/4n^2$ then $P_{ji}(t)\geq\frac{1}{2n}$ for all $i,j\in V$. Hence the graph $K_t$ on $V$, in which each edge $ij$ has weight $P_{ji}(t)+P_{ij}(t)$, has edge-expansion $\frac{1}{2}$. We finish by showing that $K_t$ embeds into $G_t$ with congestion $1/r$. Proof by induction: Consider the transition from $G_{t-1}$ to $G_t$, which is unifying $H_t$ into $G_{t-1}$. Let $i,j\in V$ be connected with an edge in $H_t$, and let $k$ be any vertex. In the transition from $K_{t-1}$ to $K_t$, we need to ship $\frac{1}{2r}$ of the type-$k$ probability in $i$ (namely $\frac{1}{2r}P_{ik}$) to $j$, and similarly, ship $\frac{1}{2r}P_{jk}$ probability from $j$ to $i$. (The ``type-$k$'' probabiility is probability mass that was originally located in $k$.) In total, we need to ship $\frac{1}{2r}\sum_{k\in V}P_{ik}=\frac{1}{2r}$ from $i$ to $j$ and a similar amount from $j$ to $i$. In total the edge $ij$ in $H_t$ needs to support $\frac{1}{r}$ flow (of probability) in the transition, so the claim follows.
\end{proof}

We turn to analyzing the change in potential in a single fixed round $t$. To simplify notation we let
\[ P_{ji} = P_{ji}(t) \;\;\;\; ; \;\;\;\; Q_{ji} = P_{ji}(t+1) . \]
Moreover recall we have a vector $u$ generated by the cut player in the current round:
\[ u=M_tM_{t-1}\ldots M_1z. \]
Denote its entries by $u_1,\ldots,u_n$. We are now adding the graph $H_{t+1}$ to $G_t$ to produce $G_{t+1}$.

\begin{mylemma}
For every $i$, $u_i$ is the projection of $P_i$ on $r$, i.e.~$u_i=P_i^Tz$.
\end{mylemma}
\begin{proof}
Fix $i$. Abbreviate $M=M_tM_{t-1}\ldots M_1\mathbf(\frac{1}{n}\mathbf1)$. If $\phi$ is any distribution on the vertices then $P_i^T\phi$ is the probability that the random walk lands in vertex $i$ after $t$ steps, meaning
\begin{equation}\label{eq:krv_projection}
(M\phi)_i=P_i^T\phi.
\end{equation}
Let $z'=\frac{1}{n\norm{z}_\infty}z$. Applying \Cref{eq:krv_projection} with $\phi=z'+\frac{1}{n}\mathbf1$ gives $(M(z'+\frac{1}{n}\mathbf1))_i=P_i^T(z'+\frac{1}{n}\mathbf1)$. Applying \Cref{eq:krv_projection} again with $\phi=\frac{1}{n}\mathbf1$ gives $(M\frac{1}{n}\mathbf1)_i=P_i^T(\frac{1}{n}\mathbf1)$ and together we get $(Mz')_i=P_i^Tz'$, which implies $u_i=(Mz)_i=P_i^Tz$.
\end{proof}

\begin{mylemma}\label{lmm:krv_random_projection}
With probability $1-1/n^{\Omega(1)}$ over the choice of $z$, for all pairs $i,j\in V$,
\[ \norm{P_i-P_j}_2^2 \geq \frac{n-1}{C\log n}|u_i-u_j|^2 . \]
\end{mylemma}
\begin{proof}
Similar to \cite[Lemma 3.4]{KRV09}.
\end{proof}

\begin{mylemma}\label{lmm:krv_cut_gain}
Let $E(S,\bar S)$ denote the set of edges in $H_{t+1}$ that cross the bisection $(S,\bar S)$ produced by the cut player (from the vector $u$). Then,
\[ (n-1)\E\left[\sum_{ij\in E(S,\bar S)}|u_i-u_j|^2\right] \geq \Psi(t) . \]
\end{mylemma}
\begin{proof}
Denote by $\text{deg}_{(S,\bar S)}(i)$ the number of edges in $E(S,\bar S)$ incident to vertex $i$. Note that $\text{deg}_{(S,\bar S)}(i)\geq 1$ for every $i\in V$, since $H_{t+1}$ is a weave on $(S,\bar S)$. Recall that $S$ contains the vertices with smallest entries in $u$. Hence there is a number $\eta\in\R$ such that $i\leq\eta\leq j$ for each edge $ij\in E(S,\bar S)$. Hence,
\begin{align*}
 \sum_{ij\in E(S,\bar S)}|u_i-u_j|^2 & \geq \sum_{ij\in E(S,\bar S)}((u_i-\eta)^2+(\eta-u_j)^2)  \\
& = \sum_{i\in V}\text{deg}_{(S,\bar S)}(i)(u_i-\eta)^2 \\
& \geq \sum_{i\in V}(u_i-\eta)^2  \\
& = \sum_{i\in V}u_i^2- 2\eta\sum_{i\in V}u_i+n\eta^2  \\
& \geq \sum_{i\in V}u_i^2, 
\end{align*}
where the last equality is by noting that $z\perp\mathbf1$, hence $u\perp\mathbf1$, hence $\sum_{i}u_i=0$.

Next, since $u_i=P_i^Tz$ and $z\perp\mathbf1$ we have $u_i=(P_i-\mathbf1/n)^Tz$. Hence $u_i$ is the projection of $P_i-\mathbf1/n$ on $z$. By properties of random projections we have $\E[u_i^2]=\frac{1}{n-1}\norm{P_i-\mathbf1/n}_2^2$ (see details in \cite{KRV09}), hence
\[
  \E\left[\sum_{i\in V}u_i^2\right] = 
  \frac{1}{n-1}\sum_{i\in V}\norm{P_i-\mathbf1/n}_2^2 =
  \frac{1}{n-1}\Psi(t),
\]
and the lemma follows from combining this with the above.
\end{proof}

\begin{mylemma}\label{lmm:krv_potential_reduction}
Let $E_{t+1}$ denote the edge set of $H_{t+1}$. The potential reduction is
\[ \Psi(t)-\Psi(t+1) =  \frac{1}{r}\sum_{ij\in E_{t+1}}\norm{P_i-P_j}_2^2 . \]
\end{mylemma}
\begin{proof}
We construct from $G$ a graph $G'$ by splitting each vertex $i$ into $r$ copies $i_1,\ldots,i_r$, assigning arbitrarily one edge from the $r$ edges incident to $i$ in $E_{t+1}$ to the copies, and distributing the type-$j$ probability in $i$, for each $j$, evenly among the copies. We denote by $P_{ji_k}$ the amount of type-$j$ probability on $i_k$ before adding $E_{t+1}$ to $G'$, and by $Q_{ji_k}$ the type-$j$ probability in $i$ after adding $E_{t+1}$. Note that we have defined $P_{ji_k}=\frac{1}{r}P_{ji}$ for all $i,j\in V$ and $k\in[r]$, but for the $Q_{ji_k}$'s all we know is that $\sum_{k=1}^rQ_{ji_k}=Q_{ji}$, so $Q_{ji}$ may be distributed arbitrarily among the $Q_{ji_k}$'s. As usual $P_{i_k}$ denotes the vector with entries $P_{ji_k}$, and $Q_{i_k}$ is defined similarly.

Define the potential of $G'$ as:
\[ \Psi'(t) = \sum_{i\in V}\sum_{k=1}^r\norm{P_{i_k}-\mathbf1/nr}_2^2. \]
We thus have
\[ \Psi(t) = 
  \sum_{i\in V}\norm{P_i-\mathbf1/n}_2^2 =
  r\sum_{k=1}^r\sum_{i\in V}\norm{\frac{1}{r}P_i-\mathbf1/nr}_2^2 =
  r\sum_{k=1}^r\sum_{i\in V}\norm{P_{i_k}-\mathbf1/nr}_2^2 =
  r\Psi'(t).
\]
To relate $\Psi(t+1)$ to $\Psi'(t+1)$, we use the general fact that for any constants $c$ and $X$, the solution to $\min\norm{x-c\mathbf1}$ s.t.~$x\in\R^r$, $\sum_ix_i=X$ is attained on $x=\frac{X}{r}\mathbf1$. Since we have $\sum_{k=1}^rQ_{ji_k}=Q_{ji}$ for all $i,j$, we infer
\begin{align*}
\Psi(t+1)  &= \sum_{i\in V}\norm{Q_i-\mathbf1/n}_2^2  \\
& = \sum_{i,j\in V}(Q_{ji}-1/n)^2 & \\
& = \sum_{i,j\in V}r\sum_{k=1}^r(\frac{1}{r}Q_{ji}-1/nr)^2 &\\
& \leq \sum_{i,j\in V}r\sum_{k=1}^r(Q_{ji_k}-1/nr)^2  \\
& = r\sum_{i\in V}\sum_{k=1}^r\norm{Q_{i_k}-\mathbf1/nr}_2^2  \\
& = r\Psi'(t+1). 
\end{align*}
We have thus proven,
\[ \Psi(t)-\Psi(t+1) \geq r(\Psi'(t)-\Psi'(t+1)). \]
Now observe that $E_{t+1}$ is, by construction, a perfect matching on $G'$. Therefore by \cite[Lemma 3.3]{KRV09} (which the current lemma generalizes),
\begin{align*}
 \Psi'(t)-\Psi'(t+1) & \geq \sum_{i_k,j_{k'}\in E_{t+1}}\norm{P_{i_k}-P_{j_{k'}}}_2^2 \\
& =  \sum_{i_k,j_{k'}\in E_{t+1}}\norm{\frac{1}{r}P_i-\frac{1}{r}P_j}_2^2 &\\
& =  \frac{1}{r^2}\sum_{i,j\in E_{t+1}}\norm{P_i-P_j}_2^2, 
\end{align*}
and the lemma follows.
\end{proof}

\begin{proof}[Proof of \Cref{thm:krv_extended}]
The initial potential is $\Psi(0)=n-1$, and by \Cref{lmm:krv_potential_to_expansion} we need to get it below $1/4n^2$. Putting \Cref{lmm:krv_random_projection,lmm:krv_cut_gain,lmm:krv_potential_reduction} together, we see that in each step we have in expectation $\Psi(t+1)\leq(1-\frac{1}{Cr\log n})\Psi(t)$. Hence, in expectation, it is enough to play for $O(r\log^2n)$ rounds.
\end{proof}

\section{Resistance Sparsification} \label{sec:resistance}

We prove \Cref{thm:main_resistance} by combining \Cref{thm:main} 
with the following known result.
\begin{theorem}[von Luxburg, Radl and Hein \cite{vLRH14}]
\label{thm:von_luxburg}
Let $G$ be a non-bipartite weighted graph with maximum edge weight $w_{\max}$ and minimum weighted degree $d_{\min}$. Let $u,v$ be nodes in $G$ with weighted degrees $d_u,d_v$ respectively. Then
\[
  \left|R_G(u,v) - \left(\frac{1}{d_u} + \frac{1}{d_v}\right)\right| \leq
  2 \left(\frac{1}{\lambda_2(G)} + 2\right) \frac{w_{\max}}{d_{\min}^2}.
\]
\end{theorem}

Qualitatively, the theorem asserts that in a sufficiently regular expander, 
the resistance distance is essentially determined by vertex degrees. Therefore an expanding subgraph $H$ of $G$ with the \emph{same} weighted degrees can serve as a resistance sparsifier. In particular, in order to resistance-sparsify a regular expander, all we need is a regular expanding subgraph, as we have by \Cref{thm:main}. Since \Cref{thm:von_luxburg} does not apply to bipartite graphs, we will use the following variant that holds also for bipartite graphs as long as they are regular.
Its proof appears in \Cref{sec:von_luxburg_bipartite}.

\begin{theorem}\label{thm:von_luxburg_bipartite}
Let $G$ be a weighted graph which is $d$-regular in weighted degrees, with maximum edge weight $w_{\max}$. Let $u,v$ be nodes in $G$. Then
\[
  \left|R_G(u,v) - \frac{2}{d}\right| \leq
  12 \left(\frac{1}{\lambda_2(G)} + 2\right) \frac{w_{\max}}{d^2}.
\]
\end{theorem}

\begin{proof}[Proof of \Cref{thm:main_resistance}]
Using \Cref{thm:main} we obtain a $d$-regular subgraph $H$ of $G$ with $\phi(H)>\frac{1}{3}$. 
By removing the obtained subgraph $H$ from $G$ and iterating,
we can apply the theorem $3d/\epsilon$ times and obtain disjoint subgraphs $H$.
Since $d=(\log n)^{O(1)}$ and $D=\Omega(n)$, the degree of $G$ does not significantly change in the process, and the requirements of \Cref{thm:main} continue to hold throughout the iterations (with a loss only in constants). Taking the union of the disjoint subgraphs produced in this process, we obtain a subgraph $H$ of $G$ which is $(3d^2/\epsilon)$-regular with $\phi(H)\geq d/\epsilon$. 
By the discrete Cheeger inequality,
\[
  \lambda_2(H) \geq
  \frac{1}{2}\left(\frac{\phi(H)}{\mathrm{deg(H)}}\right)^2 \geq
  \frac{1}{18d^2}.
\]

Recall that $H$ has edge weights in $\{1,2\}$. We now multiply each weight by $\epsilon D/(3d^2)$, rendering it $D$-regular in weighted degrees. This does not affect $\lambda_2(H)$ since it is an eigenvalue of the \emph{normalized} Laplacian.

Let $u,v\in V$. Apply \Cref{thm:von_luxburg_bipartite} on both $G$ and $H$. As $G$ is $D$-regular with $w_{\max}=1$ and $\lambda_2(G)=\Omega(1)$, we know that
$R_G(u,v) = \tfrac{2}{D} \pm O\left(\frac{1}{D^2}\right)$.
And as $H$ is $D$-regular with $w_{\max}=O(\frac{\epsilon D}{d^2})$ and $\lambda_2(H)=\Omega(1/d^2)$, we know that
$R_H(u,v) = \tfrac{2}{D} \pm O\left(\frac{\epsilon}{D}\right)$.
Putting these together, we get
$
  \frac{R_H(u,v)}{R_G(u,v)} 
  = 1\pm O\left(\epsilon+\frac{1}{D}\right) 
  = 1\pm O\left(\epsilon\right),
$
where the last equality holds for sufficiently large $n$ since $D=\Omega(n)$. Scaling $\epsilon$ down by the constant hidden in the last $O(\epsilon)$ notation yields the theorem.
\end{proof}

\subsection*{Acknowledgements}
We thank Uriel Feige for useful comments on this work.

\ifprocs
\bibliographystyle{abbrvurl}
\bibliography{resistance_lipics}
\else
\bibliographystyle{alphaurlinit}
\bibliography{resistance_lipics}
\fi

\appendix

\section{Appendix: Omitted Proofs}
\subsection{Proof of \Cref{thm:von_luxburg_bipartite}}
\label{sec:von_luxburg_bipartite}
In the non-bipartite case, \Cref{thm:von_luxburg_bipartite} follows from \Cref{thm:von_luxburg}. We henceforth assume that $G=(V,E,w)$ is bipartite with bipartition $V=V_1\cup V_2$. Note that since it is regular, we must have $|V_1|=|V_2|=\frac{1}{2}|V|$. Furthermore, as a weighted regular bipartite graph, $G$ is a convex combination of perfect matchings and hence is regular also in unweighed degrees. Let $d'$ denote the unweighted degree of each vertex in $G$. If $d'\leq2$ then it is easy to verify that the theorem holds (due to poor expansion), so we henceforth assume $d'\geq3$.

For brevity we denote the error term in \Cref{thm:von_luxburg} as 
\[ 
  \mathrm{err} 
  \eqdef 2 \left(\frac{1}{\lambda_2(G)} + 2\right) \frac{w_{\max}}{d^2}. 
\]
We will use the notion of \emph{hitting time}: For a pair of vertices $u,v$, the hitting time $H_G(u,v)$ is defined as the expected time it takes a random walk in $G$ that starts at $u$, to hit $v$. Define the \emph{normalized hitting time} $h_G(u,v)=\frac{1}{2W}H_G(u,v)$, where $W$ is the sum of all edge weights in $G$. We then have,
\begin{equation}\label{eq:resistance_hitting_time}
R_G(u,v) = h_G(u,v) + h_G(v,u).
\end{equation}
We will use the following bound on the normalized hitting time, which is given in the same theorem by von Luxburg, Radl and Hein~\cite{vLRH14}.
\begin{theorem}\label{thm:von_luxburg_hitting_time}
In the same setting of \Cref{thm:von_luxburg}, 
\[ 
  \forall u\neq v\in V,
  \qquad h_G(u,v) = \frac{1}{d_v}\pm \mathrm{err} . 
\]
\end{theorem}
(Like \Cref{thm:von_luxburg}, this theorem does not apply to bipartite graphs, and this is the obstacle we are now trying to circumvent.)

We begin by handling pairs of vertices contained within the same partition side, say $V_1$. We construct from $G$ a weighted graph $G_1$ on the vertex set $V_1$, with weights $w_1$, by putting
\[ \forall i\neq j\in V_1,\;\;\;\; w_1(i,j) = \frac{1}{d}\sum_{k\in V_2}w(i,k)w(j,k). \]
We argue that $H_{G_1}(u,v)=\frac{1}{2}H_G(u,v)$. This follows by observing that we set the weights $w_1$ such that for any $i,j\in V_1$, the probability to walk in one step from $i$ to $j$ in $G_1$ equals the probability to walk in two steps from $i$ to $j$ in $G$ via an intermediate node in $V_2$. Furthermore, we have normalized the weights $w_1$ such that $G_1$ is $d$-regular in weighted degrees. Recalling that $|V_1|=\frac{1}{2}|V|$, we have
\[ h_{G_1}(u,v)=\frac{1}{d|V_1|}H_{G_1}(u,v) = \frac{2}{d|V|}\cdot\frac{1}{2}H_G(u,v) = h_G(u,v). \]
Recalling that the unweighted degree in $G$ is $d'\geq3$, we see that by construction, $G_1$ contains a triangle and hence is non-bipartite. Hence we can apply to it \Cref{thm:von_luxburg_hitting_time} and obtain $h_{G_1}(u,v)=\frac{1}{d}\pm\mathrm{err}_1$, where $\mathrm{err}_1$ is the error term of $G_1$. Note that for every $i\neq j\in V_1$ we have $w_1(i,j)\leq\frac{w_{\mathrm{max}}}{d}\sum_{k\in V_2}w(i,k)=w_{\mathrm{max}}$, so the maximum edge weight in $G_1$ is bounded by $w_{\mathrm{max}}$, and $\lambda_2(G_1)\geq\lambda_2(G)$ (easy to verify by construction), so $\mathrm{err}_1\leq\mathrm{err}$, and we have $h_{G_1}(u,v)=\frac{1}{d}\pm\mathrm{err}$. Hence,
\[ h_G(u,v)=\frac{1}{d}\pm\mathrm{err} . \]
Recalling that $R_G(u,v) = h_G(u,v) + h_G(v,u)$, we have established that
\[ R_G(u,v) = \frac{2}{d}\pm2\mathrm{err} \]
for every pair $u,v\in V_1$. The same arguments hold for every pair $u,v\in V_2$ as well. We are left to handle the case $u\in V_1$, $v\in V_2$. Recalling the definition of hitting time, we have
\begin{align*}
H_G(u,v) &= 1+\frac{w(u,v)}{d}\cdot0+\sum_{x\in V_2\setminus\{v\}}\frac{w(u,x)}{d}H_G(x,v) & \text{(factoring out the first step)} &\\
&= 1+\frac{w(u,v)}{d}\cdot0+\sum_{x\in V_2\setminus\{v\}}\frac{w(u,x)}{d}\cdot 2W\cdot h_G(x,v) & &\\
&= 1+2W\sum_{x\in V_2\setminus\{v\}}\frac{w(u,x)}{d}\left(\frac{1}{d}\pm\mathrm{err}\right) & \text{(since $v,x\in V_2$)} &\\
&= 1+2W\left(1-\frac{w(u,v)}{d}\right)\left(\frac{1}{d}\pm\mathrm{err}\right).
\end{align*}
Therefore
\[ h_G(u,v) = \frac{1}{2W}+\left(1-\frac{w(u,v)}{d}\right)\left(\frac{1}{d}\pm\mathrm{err}\right), \]
which implies
\[ h_G \leq \frac{1}{2W}+\frac{1}{d}\pm\mathrm{err} \]
and
\[
  h_G(u,v) \geq
  \frac{1}{2W}+\left(1-\frac{w_{\max}}{d}\right)\left(\frac{1}{d}\pm\mathrm{err}\right) =
  \frac{1}{2W}+\frac{1}{d}\pm2\mathrm{err}.
\]
Together, $h_G(u,v)=\frac{1}{d}+\frac{1}{2W}\pm2\mathrm{err}$. Now, since for an arbitrary vertex $i$ we have
\[ d = \mathrm{deg}(i) = \sum_{j\in V}w(i,j) \leq nw_{\max}, \]
we see that $\frac{1}{2W}=\frac{1}{nd}\leq\frac{w_{\max}}{d^2} \leq \mathrm{err}$ and hence
\[ h_G(u,v)=\frac{1}{d}\pm3\mathrm{err}. \]
Plugging this into $R_G(u,v) = h_G(u,v) + h_G(v,u)$, we find
\[ R_G(u,v)=\frac{2}{d}\pm6\mathrm{err}, \]
which completes the proof of \Cref{thm:von_luxburg_bipartite}. \qed

\subsection{Further Omitted Proofs}\label{sec:omitted_proofs}
\begin{proposition}\label{prp:hall_consequence}
Let $G(V,U;E)$ be a bipartite graph on $n$ nodes with $|V|=|U|=\frac{1}{2}n$, and minimum degree $\geq\frac{1}{4}n$. Then $G$ contains a perfect matching.
\end{proposition}
\begin{proof}
Let $S\subset V$ be non-empty, and denote $N(S)\subset U$ the set of nodes with a neighbor in $S$. If $|S|\leq\frac{1}{4}n$ then since any $v\in S$ has $\frac{1}{4}n$ neighbors in $U$, we have $|N(S)|\geq N(\{v\})\geq\frac{1}{4}n\geq|S|$. If $|S|>\frac{1}{4}n$ then by the minimum degree condition on side $U$, every $u\in U$ must have a neighbor in $S$, and hence $|N(S)|=|U|=|V|\geq|S|$. The same arguments apply for $S\subset U$, so the condition of Hall's Marriage Theorem is verified, and it implies that $G$ contains a perfect matching.
\end{proof}

\begin{proposition}\label{prp:dense_set_cover}
Consider an instance of Set Cover with a set $S$ of $n$ elements, and a family $\mathcal M$ of subsets of $S$. Suppose each $x\in S$ belongs to at least a $\mu$-fraction of the subsets in $\mathcal M$. Then for sufficiently large $n$, we can efficiently find a cover $M\subset\mathcal M$ with $|M|\leq\frac{1.1}{\mu}\log n$.
\end{proposition}
\begin{proof}
Pick $q$ uniformly random sets (with replacement) from $\mathcal M$ to form $M$. The probability that a given element in $S$ is not covered by $M$ is upper-bounded by $(1-\mu)^q$. Taking a union bound over the element, we need to ensure that $n(1-\mu)^q<1$ in order to ensure that with constant probability, $M$ is a solution to the given Set Cover instance. This can be achieved by $q\leq\frac{1.1}{\mu}\log n$.
\end{proof}

\end{document}